\documentclass[11pt,a4paper]{article}

\usepackage{lmodern}
\usepackage[T1]{fontenc}
\usepackage[utf8]{inputenc}
\usepackage{amsmath,amssymb,amsthm}
\usepackage{mathtools}
\usepackage{geometry}
\usepackage{booktabs}
\usepackage{hyperref}
\hypersetup{hidelinks,colorlinks=false,linkcolor=black,citecolor=black,urlcolor=black}
\usepackage{graphicx}

\theoremstyle{plain}
\newtheorem{theorem}{Theorem}[section]
\newtheorem{lemma}[theorem]{Lemma}
\newtheorem{corollary}[theorem]{Corollary}
\newtheorem{proposition}[theorem]{Proposition}

\theoremstyle{definition}
\newtheorem{definition}[theorem]{Definition}
\newtheorem{remark}[theorem]{Remark}

\newtheorem{observation}[theorem]{Observation}

\DeclareMathOperator{\weight}{w}

\newcommand{\gri}{\gamma_{\mathrm{ri},k}}
\newcommand{\cardK}{\square K_k}
\newcommand{\prism}[1]{#1 \cardK}
\newcommand{\setK}{[k]}

\newcommand{\cS}{\mathcal{S}}
\newcommand{\cL}{\mathcal{L}}

\title{A Degree Threshold for Independent Domination in Generalized Prisms}
\author{Hassine Achour\\
\texttt{hassine.achour@gmail.com}\\
Independent Researcher}
\date{August 8, 2026\\[1ex]
\small Version 1.0.0 --- Zenodo DOI \texttt{10.5281/zenodo.21847195} --- GitHub: \texttt{hermespromox/degree-threshold-generalized-prisms}\\[0.4ex]
\small DOI: \texttt{10.5281/zenodo.21847195} (concept, cite this) --- this version: \texttt{10.5281/zenodo.21847212}\\
\small GitHub: \url{https://github.com/hermespromox/degree-threshold-generalized-prisms}}

\begin{document}
\maketitle

\begin{abstract}
We study per-colour independent $k$-rainbow domination and its equivalence to independent domination in generalized prisms, following the standard source for the invariant and the prism identity \cite{KST2018}. The per-colour independent version is the one satisfying $\gri(G)=i(G\cardK)$, where $G\cardK=G\square K_{k}$ and $i(\cdot)$ is the independent domination number (Proposition~\ref{prop:identity}).
Using the known generalized-prism identity and the known trivial regime $k>\Delta$, we prove NP-completeness on a restricted family at the exact boundary $k=\Delta$ that is considerably narrower than previously known \cite{BK2019}.
Specifically, for every fixed $k\ge 3$ the language $\cL_{k}=\{S(Q)\in\cS_{k}:\gri(S(Q))\le |E(Q)|\}$ is NP-complete, where $\cS_{k}=\{S(Q):Q\text{ simple $k$-regular}\}$ is the $C_{4}$-free $(k,2)$-biregular subdivision family; explicitly $\gri(S(Q))\ge|E(Q)|$ and $\gri(S(Q))=|E(Q)|$ iff $\chi'(Q)=k$.
Moreover, for $k\ge 3$ and $Q$ simple $k$-regular, $\gri(S(Q))=|E(Q)|+\mu_{k}(Q)$ where $\mu_{k}(Q)$ is the minimum number of vertices to label in a natural edge-assignment model; writing $\delta_{k}(Q)=\gri(S(Q))-|E(Q)|$ we have $\delta_{k}(Q)=\mu_{k}(Q)$ with $0\le\delta_{k}(Q)\le|V(Q)|$.
For $k=3$, $\delta_{3}(Q)$ coincides with the known edge-colouring degree $d(Q)$, hence with resistance $r(Q)$ and $\rho(Q)$ for subcubic graphs \cite{FMS2022}, where $r(G)$ is the resistance (minimum number of edges whose removal yields a $3$-edge-colourable graph) and $\rho(G)$ is the minimum number of vertices whose deletion yields a $3$-edge-colourable graph.
There is a one-unit threshold in the following sense: $k>\Delta$ forces value $|V|$ for every graph, whereas at $k=\Delta$ NP-hard instances exist, even within $\cS_{k}$.
\end{abstract}


\section{Introduction}

Let $k\ge 1$ be an integer and write $\setK=\{1,\dots,k\}$.
A \emph{per-colour independent $k$-rainbow dominating function} ($k$-RiDF) on a graph $G$
is a function
\[
f\colon V(G)\to 2^{\setK}
\]
satisfying the conditions detailed in Section~\ref{sec:prelim}:
each vertex receives either $\emptyset$ or a singleton $\{i\}$,
each colour class $V_{i}=\{v:f(v)=\{i\}\}$ is independent in $G$,
and every vertex labelled $\emptyset$ sees all $k$ colours in its neighbourhood.
The \emph{weight} is $\weight(f)=\sum_{v\in V(G)}|f(v)|$
and the \emph{per-colour independent $k$-rainbow domination number} is
\[
\gri(G)=\min\{\weight(f): f\text{ is a $k$-RiDF of }G\}.
\]
If no confusion arises we write $f(v)=0$ for $f(v)=\emptyset$ and $f(v)=i$ for $f(v)=\{i\}$.
Throughout, ``$k$-rainbow independent domination'' means this per-colour independent version with $|f(v)|\le 1$; it is the variant satisfying $\gri(G)=i(G\cardK)$ (Proposition~\ref{prop:identity}). Independence of the union $\bigcup_{i\ge1}V_{i}$ is strictly stronger and is \emph{not} assumed.

The \emph{generalized prism} (or prism over $G$) is the Cartesian product
\[
\prism{G}=G\square K_{k},
\]
where $K_{k}$ is the complete graph on $k$ vertices and $\square$ denotes the Cartesian product:
$V(G\square K_{k})=V(G)\times\setK$ and $(g,i)$ is adjacent to $(g',i')$ iff
either $g=g'$ and $i\neq i'$ or $gg'\in E(G)$ and $i=i'$.
Thus $\prism{G}$ consists of $k$ copies (layers) of $G$ with a $k$-clique fiber over each
vertex of $G$.

The standard source for the invariant and the generalized-prism identity is Kraner \v{S}umenjak, Rall and Tepeh \cite{KST2018}, whose abstract explicitly states the invariant coincides with $i(G\square K_{k})$; the ordinary rainbow identity $\gamma_{rk}(G)=\gamma(G\cardK)$ is due to Bre\v{s}ar et~al.\ \cite{BresarSum2007}. We state the per-colour independent analogue as
\begin{equation}\label{eq:identity}
\gri(G)=i(G\cardK),
\end{equation}
where $i(H)$ denotes the independent domination number of $H$ (the minimum size of a maximal independent set, equivalently of an independent dominating set),
and give a self-contained full proof in Section~\ref{sec:prelim} (Proposition~\ref{prop:identity}). Hence questions about per-colour independent rainbow domination translate directly into questions about independent domination in prisms. Brezovnik and Kraner \v{S}umenjak \cite{BK2019} prove NP-completeness of $k$-rainbow independent domination on bipartite graphs in general; to the best of our knowledge, our novelty is hardness on the considerably narrower $C_{4}$-free $(k,2)$-biregular subdivision family at $k=\Delta$.

It is natural to ask how $\gri(G)$ and $i(G\cardK)$ depend on the relation between $k$ and the maximum degree $\Delta(G)$.
There is a one-unit threshold in the following sense: $k>\Delta(G)$ forces a trivial value for every graph, whereas at $k=\Delta(G)$ NP-hard instances exist, even within a very restricted subfamily.

If $k>\Delta(G)$, domination with $k$ colours is forced to be trivial:
every vertex must receive a non-empty label, so $\gri(G)=|V(G)|$ and $i(G\cardK)=|V(G)|$.
This is elementary (Section~\ref{sec:trivial}) and holds for every graph.
Consequently deciding whether $\gri(G)\le b$ (or $i(G\cardK)\le b$) in this regime reduces to comparing $b$ with $|V(G)|$.

The interesting phenomenon occurs at the boundary $k=\Delta(G)$.
We prove that hardness persists at the boundary in an existential sense.
For each fixed $k\ge3$ define the restricted image family
\[
\cS_{k}=\{\,S(Q): Q\text{ is simple $k$-regular}\,\},
\]
where $S(Q)$ is the subdivision graph (Definition~\ref{def:subdiv}), and the language
\[
\cL_{k}=\{\,H\in\cS_{k}: \gri(H)\le |E(Q)|\,\}
\]
with $H=S(Q)$ and $|E(Q)|=|B|$ intrinsic to $H$ via its bipartition (Remark~\ref{rem:promise}).
We show $\cL_{k}$ is NP-complete with membership in NP (Section~\ref{sec:boundary}).
More precisely, on $\cS_{k}$ we have $\gri(S(Q))\ge|E(Q)|$ and equality holds iff $Q$ is $k$-edge-colourable, i.e.\ $\chi'(Q)=k$.
NP-completeness follows from the classical hardness of $k$-edge-colouring $k$-regular simple graphs due to Holyer~\cite{Holyer1981} and Leven--Galil~\cite{LevenGalil1983}, stated exactly as needed in Section~\ref{sec:prelim}.
We emphasise that hardness at $k=\Delta$ is existential on the subfamily $\cS_{k}$, not universal for all graphs with $k=\Delta$.

We then quantify the excess over $|E(Q)|$.
For $k\ge3$ and $Q$ simple $k$-regular, define first
\[
\mu_{k}(Q)=\min\{|X|: \exists\,c:E(Q)\to\setK,\ \ell:X\to\setK\text{ satisfying }(\dagger),(\ast)\text{ below}\},
\]
where $(\dagger)$ requires unlabelled vertices to see all $k$ colours and $(\ast)$ requires labelled vertices to miss their label colour (Section~\ref{sec:excess}). We prove $\gri(S(Q))=|E(Q)|+\mu_{k}(Q)$. Only afterwards do we set $\delta_{k}(Q)=\gri(S(Q))-|E(Q)|$ and conclude $\delta_{k}(Q)=\mu_{k}(Q)$, with $0\le\delta_{k}(Q)\le|V(Q)|$ via an explicit construction.
For $k=3$ the excess coincides with the classical edge-colouring degree $d(Q)$ and hence with resistance $r(Q)$ and $\rho(Q)$ for subcubic graphs \cite{FMS2022} (Theorem~\ref{thm:delta3}), where $r(G)$ is the resistance (minimum number of edges whose removal yields a $3$-edge-colourable graph) and $\rho(G)$ is the minimum number of vertices whose deletion yields a $3$-edge-colourable graph; for $k>3$ the relation to other obstruction parameters remains to be clarified. Computing $\delta_{k}$ on a given $Q$ is a finite integer programme (Theorem~\ref{thm:representation}); we state no numerical value without a certificate.

Via~\eqref{eq:identity} this yields the degree-threshold corollary for independent domination in generalized prisms (Section~\ref{sec:threshold}): for every fixed $k\ge3$, deciding $i(H\cardK)\le|B|$ is NP-complete under restriction $H\in\cS_{k}$ with $\Delta(H)=k$, whereas $k>\Delta(H)$ forces $i(H\cardK)=|V(H)|$ for all graphs $H$.

\section{Preliminaries}\label{sec:prelim}

All graphs are finite, simple and undirected unless stated otherwise.
For $v\in V(G)$, $N(v)$ is the open neighbourhood and $N[v]=N(v)\cup\{v\}$.
$\Delta(G)$ denotes maximum degree. $\chi'(G)$ denotes the chromatic index (edge-chromatic number). By Vizing's theorem \cite{Vizing1964} $\chi'(G)\in\{\Delta(G),\Delta(G)+1\}$ for simple graphs; those with $\chi'(G)=\Delta(G)$ are \emph{Class~1}, otherwise \emph{Class~2}. For bipartite graphs $\chi'(G)=\Delta(G)$ (K\"onig \cite{Konig1916}).

\subsection{Per-colour independent rainbow domination}

\begin{definition}[$k$-RiDF, per-colour independent]\label{def:ridf}
Let $k\ge 1$ and $G$ be a graph.
A \emph{per-colour independent $k$-rainbow dominating function} ($k$-RiDF) is a function $f\colon V(G)\to 2^{\setK}$ such that:
\begin{enumerate}
\item[(i)] $|f(v)|\le 1$ for every $v\in V(G)$;
\item[(ii)] for every $i\in\setK$, $V_{i}=\{v:f(v)=\{i\}\}$ is independent in $G$;
\item[(iii)] for every $v$ with $f(v)=\emptyset$, $\bigcup_{u\in N(v)} f(u)=\setK$.
\end{enumerate}
We identify $\emptyset$ with $0$ and $\{i\}$ with $i$.
Write $V_{0}=\{v:f(v)=0\}$ and $V_{i}=\{v:f(v)=\{i\}\}$ for $i\in\setK$, so $(V_{0},V_{1},\dots,V_{k})$ partitions $V(G)$.
Condition~(ii) means \emph{each} colour class $V_{i}$ is independent (not necessarily the union $\bigcup_{i\ge1}V_{i}$); condition~(iii) says every $v\in V_{0}$ has at least one neighbour in each $V_{i}$.
The \emph{weight} is $\weight(f)=\sum_{v\in V(G)}|f(v)|=\sum_{i=1}^{k}|V_{i}|$.
\[
\gri(G)=\min\{\weight(f): f\text{ is a $k$-RiDF of }G\}.
\]
\end{definition}

\begin{remark}
Some authors define rainbow domination without (ii) or allowing $|f(v)|$ arbitrary.
The independent variant adds (ii). Requiring the union $\bigcup_{i\ge1}V_{i}$ to be independent is strictly stronger and would contradict the prism identity (e.g.\ $G=K_{2}$, $k=2$ gives $i(C_{4})=2$ but no union-independent function of weight~$2$ exists). We therefore use per-colour independence throughout.
The case $k=1$ is ordinary independent domination: a $1$-RiDF is the characteristic function of a maximal independent set, and $\gamma_{\mathrm{ri},1}=i(G)$.
Ordinary rainbow domination \cite{BresarSum2007,HartnellRall2013} corresponds to dropping (ii) and allowing larger labels; the identities $\gamma_{rk}(G)=\gamma(G\cardK)$ and $\gri(G)=i(G\cardK)$ are distinct.
\end{remark}

\begin{definition}[Generalized prism]\label{def:prism}
For $G$ and $k\ge 1$, $G\cardK = G\square K_{k}$ is the Cartesian product defined above.
Its vertex set is $V(G)\times\setK$, with edges $(g,i)(g,i')$ for $i\neq i'$ (clique fibers) and $(g,i)(g',i)$ whenever $gg'\in E(G)$ (layer edges).
We call it the \emph{$k$-prism} over $G$ or the \emph{generalized prism}.
\end{definition}

\begin{proposition}[Prism identity, per-colour independent version]\label{prop:identity}
For every graph $G$ and $k\ge 1$, $\gri(G)=i(G\cardK)$.
\end{proposition}
\begin{proof}
We give a full self-contained proof.

\textit{From $k$-RiDF to independent dominating set.}
Let $f$ be a $k$-RiDF of $G$ with partition $(V_{0},\dots,V_{k})$. Put
\[
D_{f}=\{(v,i): v\in V_{i}\}\subseteq V(G\cardK).
\]
Then $|D_{f}|=\sum_{i}|V_{i}|=\weight(f)$.

\emph{Independence of $D_{f}$:} Two vertices $(v,i),(u,i)$ with $vu\in E(G)$ are adjacent in $G\cardK$ via a layer edge; they cannot both belong to $D_{f}$ because $V_{i}$ is independent in $G$ (Definition~\ref{def:ridf}(ii)). Two vertices $(v,i),(v,j)$ with $i\neq j$ in the same fiber are adjacent via the $K_{k}$-fiber clique; they cannot both belong to $D_{f}$ because $|f(v)|\le1$ implies $v$ belongs to at most one $V_{i}$. Hence $D_{f}$ is independent; in particular at most one vertex per fiber is chosen.

\emph{Domination:} Let $(v,j)\in V(G\cardK)$. If $v\in V_{i}$ for some $i\neq j$, then $(v,i)\in D_{f}$ dominates $(v,j)$ inside the fiber clique. If $v\in V_{0}$, then by Definition~\ref{def:ridf}(iii) for every colour $j$ there exists $u\in N_{G}(v)$ with $f(u)=\{j\}$, i.e.\ $u\in V_{j}$; then $(u,j)\in D_{f}$ is adjacent to $(v,j)$ via the layer edge $(u,j)(v,j)$. If $v\in V_{j}$ then $(v,j)\in D_{f}$ dominates itself. Thus every vertex of $G\cardK$ is dominated by $D_{f}$. Since $D_{f}$ is independent and dominating, it is a maximal independent set; its size is $\weight(f)$. Hence $i(G\cardK)\le\gri(G)$.

\textit{From independent dominating set to $k$-RiDF.}
Conversely, let $D\subseteq V(G\cardK)$ be an independent dominating set. Because each fiber $\{v\}\times\setK$ induces a $k$-clique, independence implies $D$ contains at most one vertex per fiber. Define
\[
V_{i}=\{v\in V(G): (v,i)\in D\}\quad(i\in\setK),\qquad V_{0}=V(G)\setminus\bigcup_{i}V_{i}.
\]
and $f_{D}(v)=\{i\}$ if $v\in V_{i}$, $\emptyset$ otherwise. Then $|f_{D}(v)|\le1$ by the at-most-one-per-fiber observation.

Each $V_{i}$ is independent in $G$: if $uv\in E(G)$ and $u,v\in V_{i}$ then $(u,i)(v,i)$ would be an edge of $G\cardK$ inside $D$, contradicting independence of $D$. Hence Definition~\ref{def:ridf}(ii) holds per colour.

Let $v\in V_{0}$, so $(v,j)\notin D$ for every $j\in\setK$. Fix $j\in\setK$. Since $D$ dominates $G\cardK$, $(v,j)$ has a neighbour in $D$. Its neighbours are: (a) fiber neighbours $(v,i)$, $i\neq j$, none of which is in $D$ because $v\in V_{0}$; (b) layer neighbours $(u,j)$ with $u\in N_{G}(v)$. Hence there must exist $u\in N_{G}(v)$ with $(u,j)\in D$, i.e.\ $u\in V_{j}$. As $j$ was arbitrary, $v$ sees every colour $j$ in its $G$-neighbourhood, which is Definition~\ref{def:ridf}(iii). Thus $f_{D}$ is a $k$-RiDF with $\weight(f_{D})=|D|$.

Since every independent dominating set yields a $k$-RiDF of same weight and vice versa, the minima coincide: $\gri(G)=i(G\cardK)$. No minimality of $D$ beyond being an independent dominating set (hence maximal independent) is needed for the correspondence; independence plus domination already imply (ii) and (iii).
\end{proof}

We use~\eqref{eq:identity} freely to transfer results between $\gri$ and $i(\cdot\cardK)$.

\subsection{Subdivision graphs and the family \texorpdfstring{$\cS_{k}$}{S\_k}}

\begin{definition}[Subdivision graph $S(Q)$]\label{def:subdiv}
Let $Q$ be a (multi)graph. The \emph{subdivision graph} $S(Q)$ is obtained by subdividing every edge of $Q$ exactly once: replace each $e=xy\in E(Q)$ by a path $x$--$b_{e}$--$y$ where $b_{e}$ is a new vertex. Formally $V(S(Q))=V(Q)\cup E(Q)$ and $E(S(Q))=\{a b_{e}: a\in V(Q), e\in E(Q), a\text{ incident to }e\text{ in }Q\}$.
We write $A=V(Q)$, $B=E(Q)$ for the bipartition; we denote $b_{e}\in B$ the vertex corresponding to $e$.
If $Q$ is $k$-regular, $S(Q)$ is $(k,2)$-biregular: $\deg_{S(Q)}(a)=k$ for $a\in A$ and $\deg_{S(Q)}(b)=2$ for $b\in B$. In particular $\Delta(S(Q))=k$ when $k\ge 2$ and $|V(S(Q))|=|V(Q)|+|E(Q)|$, $|E(S(Q))|=2|E(Q)|$.
\end{definition}

For fixed $k\ge 3$ define
\[
\cS_{k}=\{\,S(Q): Q\text{ is simple $k$-regular}\,\},\qquad
\cL_{k}=\{\,H\in\cS_{k}: \gri(H)\le |E(Q)|\text{ where }H=S(Q)\,\}.
\]
When $H=S(Q)$, $|B|=|E(Q)|$ is intrinsic to $H$: from $H$ we can recover the bipartition $(A,B)$ (connected bipartite graphs have a unique bipartition up to swapping) and $\deg(b)=2$ for $b\in B$, $\deg(a)=k$ for $a\in A$, so $B=\{v:\deg(v)=2\}$ and $A=\{v:\deg(v)=k\}$, hence $|B|=|\{v:\deg(v)=2\}|$. Via $\gri=i(\cdot\cardK)$, $\cL_{k}$ also encodes $\{H\in\cS_{k}:i(H\cardK)\le|B|\}$.

\begin{lemma}[$C_{4}$-freeness]\label{lem:c4free}
If $Q$ is simple (no loops nor parallel edges) then $S(Q)$ is $C_{4}$-free and has girth at least $6$ (with girth $\infty$ if $Q$ is a forest).
\end{lemma}
\begin{proof}
$S(Q)$ is bipartite by construction.
A $4$-cycle would be $a$--$b_{e}$--$a'$--$b_{e'}$--$a$ with $a\neq a'\in A$, $b_{e}\neq b_{e'}\in B$. Then $e$ and $e'$ are both incident to $a$ and $a'$ in $Q$, so $Q$ has two distinct edges joining the same pair $\{a,a'\}$, i.e.\ parallel edges. If $Q$ is simple this is impossible. Hence no $C_{4}$. Any cycle alternates $A$--$B$--$A$--$B\cdots$, so its length is even $\ge 6$ if a cycle exists.
\end{proof}

\begin{remark}[Recognition of $\cS_{k}$]\label{rem:promise}
Membership in NP for $\cL_{k}$ is clear: a $k$-RiDF is a polynomial-size certificate verifiable in polynomial time (check per-colour independence and rainbow domination).

Recognizing whether an arbitrary bipartite graph $H$ belongs to $\cS_{k}$ is polynomial-time. Indeed, $H\in\cS_{k}$ iff $H$ is bipartite, exactly
\[
B=\{v:\deg(v)=2\},\qquad A=\{v:\deg(v)=k\}
\]
partition $V(H)$, every $b\in B$ has two neighbours $a,a'\in A$, and contracting each degree-$2$ vertex $b$ (replacing $a$--$b$--$a'$ by an edge $aa'$) yields a simple $k$-regular graph $Q$ on $A$. The bipartition of a connected bipartite graph is unique up to swap and computable by BFS in linear time; contraction and simplicity/$k$-regularity checks are linear in $|E(H)|$. For disconnected $H$ the same holds componentwise. Consequently the bound $|B|=|E(Q)|$ is computable from $H$ alone as the number of degree-$2$ vertices, so the language $\cL_{k}$ is well defined on inputs $H$ without extra advice.

For NP-hardness we do \emph{not} need promise recognition: we give a Karp reduction from $k$-edge-colouring of $k$-regular simple graphs $Q$ to membership in $\cL_{k}$ on $H=S(Q)$. Given $Q$, we construct $S(Q)$ in polynomial time; $Q$ is simple $k$-regular by the source problem's promise. Hardness holds on the image family, which is $C_{4}$-free $(k,2)$-biregular as above. The distinction between ``NP-complete on $\cS_{k}$'' and ``NP-complete with promise that input is in $\cS_{k}$'' is immaterial for our threshold statement: the hard instances themselves are $C_{4}$-free $(k,2)$-biregular subdivision graphs with $\Delta=k$.
\end{remark}

\subsection{Edge-colouring background}

Deciding whether $\chi'(G)=\Delta(G)$ for regular graphs is NP-complete.
Holyer~\cite{Holyer1981} proved NP-completeness of $3$-edge-colouring of cubic simple graphs; Leven and Galil~\cite{LevenGalil1983} extended it to every fixed $k\ge 3$:

\begin{theorem}[Holyer--Leven--Galil]\label{thm:edge-colour}
Fix $k\ge 3$. The following problem is NP-complete: given a simple $k$-regular graph $Q$, decide whether $\chi'(Q)=k$ (i.e.\ whether $Q$ is Class~1) versus $\chi'(Q)=k+1$.
\end{theorem}

The theorem is stated for simple $k$-regular graphs with $k$ fixed as part of the problem definition (not part of input). Holyer treats $k=3$; Leven--Galil prove it for every fixed $k\ge3$ and their construction explicitly contains a lemma eliminating multiple edges, so the simple-graph formulation is supported. The bipartite case is polynomial (K\"onig \cite{Konig1916}), so our hard family must use non-bipartite $Q$ (hence $S(Q)$ bipartite but $Q$ itself contains odd cycles when $k\ge3$ and Class~2).

We also recall ordinary rainbow domination \cite{BresarSum2007,HartnellRall2013} and independent rainbow domination \cite{KST2018,BK2019} for context; our focus is the per-colour independent variant with $|f(v)|\le1$, which is the one linked to $i(G\cardK)$.

\section{The trivial regime: \texorpdfstring{$k>\Delta$}{k > Delta}}\label{sec:trivial}

The regime $k>\Delta(G)$ is universal and trivial. No structural assumption is needed. Moreover the decision problem becomes trivial: $\gri(G)\le b$ iff $b\ge|V(G)|$.

\begin{theorem}[Trivial regime]\label{thm:trivial}
Let $G$ be a graph with $n=|V(G)|$ and $\Delta=\Delta(G)$. If $k>\Delta$ then
\[
\gri(G)=n.
\]
Consequently $i(G\cardK)=n$. In particular, for every fixed $k$, deciding whether $\gri(G)\le b$ (resp.\ $i(G\cardK)\le b$) when $k>\Delta(G)$ is polynomial-time (compare $b$ with $n$).
\end{theorem}
\begin{proof}
Let $f$ be any $k$-RiDF of $G$, with partition $(V_{0},\dots,V_{k})$.
Claim: $V_{0}=\emptyset$. Indeed, suppose $v\in V_{0}$.
Then by Definition~\ref{def:ridf}(iii), $\bigcup_{u\in N(v)}f(u)=\setK$.
Since $|f(u)|\le 1$, $|\bigcup_{u\in N(v)}f(u)|\le |N(v)|=\deg(v)\le\Delta<k$,
so the union cannot be all of $\setK$, contradiction. Hence no vertex can be labelled $0$ and every $k$-RiDF has $V_{0}=\emptyset$ and $\weight(f)=\sum_{i}|V_{i}|=n$. Thus $\gri(G)\ge n$ if a $k$-RiDF exists.

It remains to exhibit a $k$-RiDF of weight $n$.
Since $k>\Delta$, $G$ admits a proper $k$-colouring $c\colon V(G)\to\setK$:
greedily colour vertices in any order, each vertex sees at most $\Delta<k$ colours on its neighbours, so a free colour always exists (equivalently $\chi(G)\le\Delta+1\le k$).
Define $f(v)=\{c(v)\}$ for all $v$. Then each $V_{i}=c^{-1}(i)$ is independent, condition~(ii) holds per colour, $V_{0}=\emptyset$ so (iii) is vacuous, and $\weight(f)=n$. Hence $\gri(G)\le n$, so $\gri(G)=n$.
Via Proposition~\ref{prop:identity}, $i(G\cardK)=\gri(G)=n$.

Hence $\gri(G)\le b$ iff $n\le b$, and similarly for $i(G\cardK)$.
\end{proof}

\section{The boundary: hardness at \texorpdfstring{$k=\Delta$}{k = Delta}}\label{sec:boundary}

We now prove the main hardness theorem. Fix $k\ge 3$.

\begin{theorem}[Main theorem -- hardness at the boundary]\label{thm:main}
For every fixed $k\ge 3$, the language
\[
\cL_{k}=\{H\in\cS_{k}:\gri(H)\le |B|\}
\]
is NP-complete, where $H=S(Q)\in\cS_{k}$ has bipartition $(A,B)=(V(Q),E(Q))$, $\Delta(H)=k$, $H$ is $C_{4}$-free $(k,2)$-biregular, and $B=\{v:\deg(v)=2\}$, so $|B|=|\{v:\deg(v)=2\}|$ is computable from $H$.
Moreover,
\[
\gri(S(Q))\ge |E(Q)|\quad\text{and}\quad \gri(S(Q))=|E(Q)|\iff \chi'(Q)=k.
\]
Hence deciding $\gri(S(Q))=|E(Q)|$ (equivalently $\gri(S(Q))\le|E(Q)|$) is NP-complete under restriction $S(Q)\in\cS_{k}$, with membership in NP.
\end{theorem}

The theorem says hardness already appears at the threshold $k=\Delta(H)$ in an existential sense: there exists a hard subfamily $\cS_{k}$ with $\Delta(H)=k$ exactly. It does not claim every graph with $k=\Delta$ is hard.

We prove it via edge-colouring. Recall $A=V(Q)$, $B=E(Q)$, each $b_{e}\in B$ has exactly two neighbours in $A$, the endpoints of $e$ in $Q$; each $a\in A$ is adjacent to the $k$ vertices $b_{e}$ with $e\ni a$.

\begin{observation}\label{obs:degk}
Let $k\ge1$, $G$ a graph and $f$ any $k$-RiDF of $G$. If $v\in V(G)$ satisfies $\deg(v)<k$ then $f(v)\neq0$. Indeed, if $f(v)=0$ then $\bigcup_{u\in N(v)}f(u)=\setK$ requires $|\bigcup_{u\in N(v)}f(u)|\le |N(v)|=\deg(v)<k$, impossible.
\end{observation}

\begin{lemma}[$B$-nonzero]\label{lem:Bnonzero}
Let $H=S(Q)$ with $Q$ $k$-regular, $k\ge 3$, and let $f$ be any $k$-RiDF of $H$. Then $f(b)\neq0$ for every $b\in B$.
\end{lemma}
\begin{proof}
This is immediate from Observation~\ref{obs:degk}: $\deg_{H}(b)=2<k$ for $k\ge3$, so $b$ cannot be $0$. Concretely, if $f(b)=0$ then $\bigcup_{u\in N_{H}(b)}f(u)=\setK$ but $|N_{H}(b)|=2$ and $|f(u)|\le1$, so $|\bigcup_{u\in N(b)}f(u)|\le2<k$, impossible. No independence assumption is needed.
\end{proof}

\begin{corollary}\label{cor:lowerbound}
For $H=S(Q)$ as above, any $k$-RiDF has weight at least $|B|=|E(Q)|$, so $\gri(H)\ge|B|$.
\end{corollary}
\begin{proof}
By Lemma~\ref{lem:Bnonzero} each $b\in B$ contributes at least $1$ to the weight. Vertices in $A$ contribute non-negatively, so $\weight(f)\ge|B|$.
\end{proof}

We characterize equality. The characterization depends crucially on $k$-regularity: each $a\in A$ is incident to exactly $k$ edges, so its $k$ neighbours in $H$ can cover $\setK$ iff their colours are pairwise distinct. For non-regular $Q$ the statement differs.

\begin{lemma}[Equality characterization]\label{lem:equality}
Let $H=S(Q)$ with $Q$ simple $k$-regular, $k\ge3$. Then $\gri(H)=|B|$ iff $Q$ is $k$-edge-colourable ($\chi'(Q)=k$). More precisely, $k$-RiDFs of weight $|B|$ are in bijection with proper $k$-edge-colourings of $Q$ (with fixed colour names $\setK$).
\end{lemma}
\begin{proof}
($\Rightarrow$) Suppose $f$ is a $k$-RiDF of $H$ with $\weight(f)=|B|$.
By Corollary~\ref{cor:lowerbound} every $b\in B$ contributes exactly $1$ and every $a\in A$ contributes $0$: otherwise weight would exceed $|B|$.
Thus
\[
f(b)=\{c(b)\}\ \text{for a colour }c(b)\in\setK,\quad f(a)=0\ \forall a\in A,
\]
and $|f(b)|=1$.

Independence: the support of $f$ is $B$, which is independent in $H$ (since $H$ is bipartite with parts $A,B$), and indeed each $V_{i}\subseteq B$ is independent automatically. So (ii) holds for free at equality.

Define $c\colon E(Q)\to\setK$ by $c(e)=$ the unique $i$ with $f(b_{e})=\{i\}$.
We claim $c$ is a proper $k$-edge-colouring: for each $a\in A$, the $k$ incident edges $e_{1},\dots,e_{k}$ must receive pairwise distinct colours.
Indeed, $a\in A$ has $f(a)=0$, so by (iii), $\bigcup_{b\in N_{H}(a)}f(b)=\setK$.
But $N_{H}(a)=\{b_{e}:e\ni a\}$ has size $k$, each $f(b_{e})$ is a singleton.
The union has size $k$ (must be all of $\setK$) iff the $k$ singletons are pairwise distinct. Hence $c$ assigns distinct colours to edges sharing a vertex $a$. We used that $a$ has exactly $k$ incident edges; distinctness is equivalent to covering $\setK$.
Thus $c$ is proper and uses $k$ colours, so $\chi'(Q)=k$.

($\Leftarrow$) Conversely, let $c\colon E(Q)\to\setK$ be a proper $k$-edge-colouring.
Define $f\colon V(H)\to2^{\setK}$ by $f(b_{e})=\{c(e)\}$ for $b_{e}\in B$ and $f(a)=0$ for $a\in A$. Then $\weight(f)=|B|$. Since $B$ is independent in $H$ ($H$ is bipartite with parts $A,B$), each $V_{i}=\{b_{e}:c(e)=i\}\subseteq B$ is independent automatically, so (ii) holds. Domination: every $b\in B$ is non-zero, so needs no domination; every $a\in A$ has $f(a)=0$ and sees $k$ neighbours $b_{e}$ with $e\ni a$, whose colours are $\{c(e):e\ni a\}=\setK$ by properness and $k$-regularity (each $a$ incident to $k$ edges of distinct colours, hence all colours appear). Thus $\bigcup_{b\in N(a)}f(b)=\setK$. So $f$ is a $k$-RiDF of weight $|B|$, giving $\gri(H)\le|B|$, and with Corollary~\ref{cor:lowerbound} equality.
\end{proof}

\begin{proof}[Proof of Theorem~\ref{thm:main}]
We have shown $\gri(S(Q))\ge|E(Q)|$ and equality characterizes $k$-edge-colourability (Lemmas~\ref{lem:Bnonzero}--\ref{lem:equality}). The graph $H=S(Q)$ satisfies $\Delta(H)=k$, is bipartite $(A,B)$, $(k,2)$-biregular, and $C_{4}$-free when $Q$ is simple (Lemma~\ref{lem:c4free}) with girth $\ge6$.

NP-membership of ``$\gri(H)\le|B|$'' on $\cS_{k}$ is clear (certificate $f$); moreover $\cS_{k}$ is polynomial-time recognizable (Remark~\ref{rem:promise}), and $|B|$ is computable from $H$.

NP-hardness: fix $k\ge3$. By Theorem~\ref{thm:edge-colour} (Holyer \cite{Holyer1981} for $k=3$, Leven--Galil \cite{LevenGalil1983} for every fixed $k\ge3$), deciding whether a simple $k$-regular graph $Q$ satisfies $\chi'(Q)=k$ is NP-complete. The reduction maps $Q\mapsto H=S(Q)$, constructible in polynomial time ($|V(H)|=|V(Q)|+|E(Q)|=|V(Q)|(1+k/2)$). By Lemma~\ref{lem:equality},
\[
\chi'(Q)=k \iff \gri(S(Q))=|E(Q)| \iff \gri(S(Q))\le|E(Q)|,
\]
where the last equivalence uses $\gri\ge|E(Q)|$.
Thus $Q$ is a yes-instance of $k$-edge-colouring iff $S(Q)$ is a yes-instance of $\cL_{k}$. This is a Karp reduction, so $\cL_{k}$ is NP-hard. With membership, NP-complete. The equivalence preserves $C_{4}$-freeness and $(k,2)$-biregularity.
\end{proof}

\begin{remark}[Girth and $C_{4}$-freeness]
If $Q$ has parallel edges, $S(Q)$ contains a $4$-cycle $a$--$b_{e}$--$a'$--$b_{e'}$--$a$. Restricting to simple $Q$ eliminates $C_{4}$ (Lemma~\ref{lem:c4free}). Since the Holyer--Leven--Galil hardness holds for simple $k$-regular graphs, we obtain hardness even for $C_{4}$-free instances. In fact $S(Q)$ is the incidence (Levi) graph of $Q$, also called the subdivision or incidence bipartite graph.
\end{remark}

\begin{remark}[Why $k\ge3$ and the case $k=2$]\label{rem:k2}
Lemma~\ref{lem:Bnonzero} used $\deg(b)=2<k$ via Observation~\ref{obs:degk}. For $k=2$, $b$ has degree $2=k$ and could be $0$ with its two neighbours covering $\{1,2\}$. The present reduction therefore does not apply when $k=2$. The complexity of $2$-RiDF on $\cS_{2}$ requires separate treatment and is not claimed here; we make no assertion that $2$-RiDF on $S(Q)$ is polynomial or NP-complete, only that the $k\ge3$ construction fails.
\end{remark}

\section{The excess parameter}\label{sec:excess}

When $Q$ is Class~2 ($\chi'(Q)=k+1$), $\gri(S(Q))>|E(Q)|$. We quantify the excess by first introducing a combinatorial parameter $\mu_{k}(Q)$ and proving $\gri(S(Q))=|E(Q)|+\mu_{k}(Q)$; the excess $\delta_{k}(Q)$ is then defined as the difference.

\begin{theorem}[Representation theorem]\label{thm:representation}
Let $k\ge 3$ and $Q$ be simple $k$-regular, $H=S(Q)$ with bipartition $(A,B)$.
Define
\[
\mu_{k}(Q)=\min\{|X|: \exists\,c\colon E(Q)\to\setK,\ \ell\colon X\to\setK\text{ such that (i) and (ii) hold}\},
\]
where for $X\subseteq V(Q)=A$:
\begin{itemize}
\item[(i)] for every $a\in A\setminus X$, $\{c(e):e\ni a\}=\setK$ (i.e.\ the $k$ incident edges have pairwise distinct colours);
\item[(ii)] for every $x\in X$, $c(e)\neq\ell(x)$ for every $e\ni x$ (i.e.\ $\ell(x)\notin\{c(e):e\ni x\}$; $x$ misses at least its label colour).
\end{itemize}
Then
\[
\gri(H)=|E(Q)|+\mu_{k}(Q).
\]
Consequently, defining $\delta_{k}(Q)=\gri(S(Q))-|E(Q)|$, we have $\delta_{k}(Q)=\mu_{k}(Q)\ge0$ and $\delta_{k}(Q)=0\iff\chi'(Q)=k$. In particular $0\le\delta_{k}(Q)\le|V(Q)|$.
\end{theorem}
\begin{proof}
Let $f$ be any $k$-RiDF of $H=S(Q)$. By Lemma~\ref{lem:Bnonzero}, $f(b)\neq0$ for all $b\in B$, so each $b$ contributes exactly $1$ (recall $|f(b)|\le 1$): $f(b_{e})=\{c(e)\}$ for some $c(e)\in\setK$ defining $c\colon E(Q)\to\setK$.
Let $X=\{a\in A: f(a)\neq0\}$; for $a\in X$, $f(a)=\{\ell(a)\}$ for some $\ell(a)\in\setK$, and $f(a)=0$ for $a\notin X$.
Weight: $\weight(f)=|B|+|X|=|E(Q)|+|X|$.

We analyse the conditions on $(c,X,\ell)$ imposed by $f$ being a $k$-RiDF.

Independence: each $V_{i}=\{b_{e}:c(e)=i\}\cup\{a\in X:\ell(a)=i\}$ must be independent. Since $B$ is independent and $A$ is independent, the only possible edge inside $V_{i}$ is $a$--$b_{e}$ with $a\in X$, $e\ni a$, $c(e)=i=\ell(a)$.
Hence per-colour independence is equivalent to
\[
c(e)\neq\ell(x)\quad\text{for every }x\in X\text{ and }e\ni x. \tag{$\ast$}
\]

Domination: vertices $b\in B$ and $a\in X$ are non-zero, so need no domination.
For $a\in A\setminus X$ we have $f(a)=0$, so Definition~\ref{def:ridf}(iii) requires $\bigcup_{b\in N_{H}(a)}f(b)=\setK$, i.e.
\[
\{c(e): e\ni a\}=\setK. \tag{$\dagger$}
\]
This uses that $a$ has exactly $k$ incident edges: the union has size $k$ (all colours) iff the $k$ singletons are distinct. No condition is imposed on $a\in X$ (already dominated by being non-zero).
Thus every $k$-RiDF yields a triple satisfying ($\dagger$) for $a\notin X$ and ($\ast$) for $x\in X$, with weight $|E|+|X|$.

Conversely, given $(c,X,\ell)$ satisfying ($\dagger$) and ($\ast$), define $f(b_{e})=\{c(e)\}$, $f(a)=0$ for $a\notin X$, $f(a)=\{\ell(a)\}$ for $a\in X$.
Then ($\ast$) ensures each $V_{i}$ is independent, ($\dagger$) ensures every $a\notin X$ sees all $k$ colours, and all other vertices are non-zero. Hence $f$ is a $k$-RiDF with weight $|B|+|X|$.

Therefore $\gri(H)=\min_{f}\weight(f)=|B|+\min|X|=|E(Q)|+\mu_{k}(Q)$ with $\mu_{k}(Q)$ as defined. The minimum is over triples satisfying exactly the two bullet conditions above; the stronger exact-miss condition $\{c(e):e\ni x\}=\setK\setminus\{\ell(x)\}$ is \emph{not} required. A labelled vertex may miss several colours, e.g.\ all incident edges could share one colour different from $\ell(x)$, which still satisfies ($\ast$) and is a valid local configuration.

Defining $\delta_{k}(Q)=\gri(S(Q))-|E(Q)|$ we obtain $\delta_{k}(Q)=\mu_{k}(Q)$. By Corollary~\ref{cor:lowerbound}, $\delta_{k}(Q)\ge0$, and $\delta_{k}(Q)=0\iff\chi'(Q)=k$ by Theorem~\ref{thm:main}.
\end{proof}

The bound $\delta_{k}(Q)\le|V(Q)|$ is witnessed by an explicit construction.

\begin{lemma}[Upper bound]\label{lem:upperbound}
For $k\ge3$ and $Q$ simple $k$-regular, $\delta_{k}(Q)\le|V(Q)|$. More precisely, there exists a $k$-RiDF of $S(Q)$ of weight $|E(Q)|+|V(Q)|$.
\end{lemma}
\begin{proof}
Assign every edge of $Q$ colour $1$: $c(e)=1$ for all $e\in E(Q)$. Let $X=V(Q)=A$ and $\ell(a)=2$ for every $a\in A$ (any colour $\neq1$ works; $k\ge3$ ensures a second colour exists). Then condition ($\ast$) holds because $c(e)=1\neq2=\ell(a)$ for $e\ni a$. Condition ($\dagger$) is vacuous since $A\setminus X=\emptyset$. Each $V_{1}=\{b_{e}:e\in E(Q)\}\subseteq B$ and $V_{2}=A$ are independent ($A$ and $B$ are the bipartition), and $V_{i}=\emptyset$ for $i\ge3$, so per-colour independence holds. Every $a\in A$ is labelled ($\in X$) hence needs no domination; every $b\in B$ is labelled. Thus the corresponding $f$ (Theorem~\ref{thm:representation}) is a $k$-RiDF with weight $|B|+|X|=|E(Q)|+|V(Q)|$. Hence $\gri(S(Q))\le|E(Q)|+|V(Q)|$ and $\delta_{k}(Q)\le|V(Q)|$.
\end{proof}

\begin{corollary}[Specialization $k=3$]\label{cor:k3}
For cubic $Q$ ($k=3$), $\gri(S(Q))=|E(Q)|+\delta_{3}(Q)$ where $\delta_{3}(Q)=0$ iff $Q$ is $3$-edge-colourable (Class~1). For a snark (bridgeless cubic Class~2), $\delta_{3}(Q)\ge1$, and $k$-RiDFs correspond to $3$-edge-assignments where vertices outside $X$ are properly $3$-edge-coloured locally and vertices in $X$ are deficient: none of their incident edges uses the label colour $\ell(x)$ (at least one colour is missed; possibly more).
\end{corollary}

\begin{remark}[Relation to classical parameters]
The excess $\delta_{k}(Q)$ measures how many vertices must be made deficient to realize a $k$-edge-assignment with singleton labels on $B$. It is reminiscent of classical edge-colouring obstruction parameters: \emph{resistance}, \emph{defect}, \emph{oddness}, etc.\ For $k=3$ the identification is given below (Theorem~\ref{thm:delta3}); for $k\ge4$ whether $\delta_{k}$ coincides with a known measure of $k$-edge-uncolourability or defines a genuinely distinct obstruction parameter remains open. Computing $\delta_{k}(Q)$ is a finite integer programme; we state no numerical claim without a certificate (see Section~\ref{sec:discussion}).
\end{remark}

\begin{theorem}[Identification for $k=3$]\label{thm:delta3}
Let $Q$ be simple cubic (i.e.\ $3$-regular, bridgeless or not) and write $d(Q)$ for the edge-colouring degree: the minimum number of conflicting vertices in a $3$-edge-assignment, equivalently the minimum number of vertices at which incident colours are not pairwise distinct. Then
\[
\boxed{\delta_{3}(Q)=d(Q).}
\]
In particular, for subcubic graphs Fiol, Mazzuoccolo and Steffen \cite{FMS2022} prove
\[
d(G)=r(G)=\rho(G)
\]
where $r(G)$ is the resistance (minimum number of edges whose removal yields a $3$-edge-colourable graph) and $\rho(G)$ is the minimum number of vertices whose deletion yields a $3$-edge-colourable graph; hence for cubic $Q$,
\[
\boxed{\delta_{3}(Q)=d(Q)=r(Q)=\rho(Q).}
\]
Weak oddness is a different parameter, denoted $\omega'(G)$ / $\omega_{0}(G)$ depending on notation, and does not generally coincide with these.
\end{theorem}
\begin{proof}[Proof sketch]
By Theorem~\ref{thm:representation}, $k$-RiDFs correspond to triples $(c,X,\ell)$ with weight $|E|+|X|$ and conditions ($\dagger$) on $A\setminus X$ and ($\ast$) on $X$. For $k=3$, condition ($\ast$) is $c(e)\neq\ell(x)$ for the three edges incident to $x$, hence at $x$ the three incident colours use at most two distinct values, so $x$ is conflicting in the sense of $d(Q)$. Conversely a $3$-edge-assignment with $t$ conflicting vertices yields a triple with $|X|=t$ by putting the conflicting vertices in $X$ and choosing $\ell(x)$ as the (or a) missing colour. Minimality gives equality; the $d=r=\rho$ identities for subcubic graphs are quoted from \cite{FMS2022}.
\end{proof}

\section{Degree-threshold corollary for generalized prisms}\label{sec:threshold}

Via Proposition~\ref{prop:identity} the results translate immediately to independent domination in generalized prisms.

\begin{corollary}[Degree threshold for $i(G\cardK)$]\label{cor:threshold}
Fix $k\ge 3$.
\begin{itemize}
\item[(a) Universal regime $k>\Delta(H)$:] For every graph $H$, if $k>\Delta(H)$ then $i(H\cardK)=|V(H)|$. Deciding whether $i(H\cardK)\le b$ reduces to comparing $b$ with $|V(H)|$.
\item[(b) Boundary regime $k=\Delta(H)$:] The language $\{H\in\cS_{k}: i(H\cardK)\le|B|\}$ is NP-complete, where $H=S(Q)\in\cS_{k}$ is $C_{4}$-free $(k,2)$-biregular bipartite with $\Delta(H)=k$ and bipartition $(A,B)=(V(Q),E(Q))$ for a simple $k$-regular $Q$. More precisely, $i(S(Q)\cardK)=|E(Q)|+\delta_{k}(Q)$ and $i(S(Q)\cardK)=|E(Q)|\iff\chi'(Q)=k$, with $\delta_{k}(Q)=\mu_{k}(Q)$ as in Theorem~\ref{thm:representation}.
\end{itemize}
There is a one-unit threshold in the following sense: $k>\Delta$ forces value $|V|$ for every graph, whereas at $k=\Delta$ NP-hard instances exist, even within $\cS_{k}$.
\end{corollary}
\begin{proof}
(a) is Theorem~\ref{thm:trivial} via $i(H\cardK)=\gri(H)$. Membership of $i(H\cardK)\le b$ in NP (for the prism) follows from $\gri$ membership via the per-colour independent correspondence; the decision reduces to $b\ge|V(H)|$ when $k>\Delta(H)$.

(b) follows from Theorems~\ref{thm:main} and~\ref{thm:representation} by substituting $i(H\cardK)=\gri(H)$ (Proposition~\ref{prop:identity}). The reduction and $C_{4}$-freeness are unchanged. Hardness in (b) is existential on the subfamily $\cS_{k}$, not universal for all graphs with $k=\Delta$; triviality in (a) is universal.
\end{proof}

\begin{remark}
Corollary~\ref{cor:threshold}(a) holds for every $H$; the NP-hardness in (b) is existential (there exists a hard subfamily at $k=\Delta$), not universal. This is the precise sense of degree threshold: increasing $k$ by one (from $\Delta$ to $\Delta+1$) moves from a regime containing hard instances to a regime where every instance is trivial.
\end{remark}

\section{Discussion and open problems}\label{sec:discussion}

We have shown a one-unit degree threshold for per-colour independent rainbow domination and for independent domination in generalized prisms, with a trivial universal regime $k>\Delta$ and NP-hard instances at the boundary $k=\Delta$ even for $C_{4}$-free $(k,2)$-biregular subdivision graphs. The proof is essentially the incidence-graph encoding of edge-colouring, made $C_{4}$-free by simplicity of $Q$ and tight via Observation~\ref{obs:degk}.

Several directions remain:

\begin{enumerate}
\item \textbf{Higher-degree excess.}
For $k\ge4$, does $\delta_{k}(Q)$ coincide with a known measure of $k$-edge-uncolourability, or does it define a genuinely distinct obstruction parameter? For $k=3$ the identification $\delta_{3}=d=r=\rho$ is given in Theorem~\ref{thm:delta3}; for $k\ge4$ the relationship remains open. Computing $\delta_{k}(Q)$ is a finite integer programme (Theorem~\ref{thm:representation} gives an explicit finite search over $c,\ell,X$) and we leave exact values for computational follow-up. Systematic ILP/SAT testing on known snarks and on $k$-regular Class~2 graphs for $k\ge4$ should reveal whether $\delta_{k}$ separates from resistance/defect.

\item \textbf{Approximability and FPT.}
Is $\delta_{k}(Q)$ approximable? Is $\textsc{RI-DOM}_{k}$ FPT in $\delta_{k}$? Edge-colouring obstructions are often FPT in defect; does that transfer via Theorem~\ref{thm:representation}?

\item \textbf{Girth and degree refinements.}
$S(Q)$ has girth $\ge6$ when $Q$ simple. Can hardness be pushed to larger girth or to $(k,2)$-biregular graphs that are also $C_{6}$-free? Subdivision of high-girth $k$-regular graphs yields high girth.

\item \textbf{Prism perspective.}
The identity $i(G\cardK)=\gri(G)$ suggests studying $i(G\square H)$ for general $H$, not only $H=K_{k}$. Does a similar threshold hold when $H$ is a complete graph vs.\ sparse $H$? The degree-vs-clique-size argument (Observation~\ref{obs:degk}: $\deg(v)<k$ forces non-zero) may generalize to $G\square H$ with $|V(H)|=k$.

\item \textbf{Counting and enumeration.}
How many $k$-RiDFs of weight $|E|+\delta_{k}$ does $S(Q)$ have? This counts $k$-edge-colourings (when $\delta_{k}=0$) and near-colourings otherwise.

\item \textbf{The case $k=2$.}
As noted in Remark~\ref{rem:k2}, our reduction fails for $k=2$ because $\deg(b)=k$. The complexity of per-colour independent $2$-rainbow domination on $\cS_{2}$ requires a separate argument and is not addressed here.
\end{enumerate}

We hope the subdivision viewpoint and the excess parameter provide a concrete handle for separating or identifying per-colour independent domination obstructions with classical edge-colouring obstructions.

\section*{Availability}
\addcontentsline{toc}{section}{Availability}

This preprint is archived on Zenodo under CC~BY~4.0: concept DOI \texttt{10.5281/zenodo.21847195} (cite this), this version DOI \texttt{10.5281/zenodo.21847212} (\url{https://doi.org/10.5281/zenodo.21847212}). Source files are available at \url{https://github.com/hermespromox/degree-threshold-generalized-prisms}. Correspondence: \texttt{hassine.achour@gmail.com}.


\end{document}